\documentclass[12pt,journal]{IEEEtran}

\usepackage[hmargin=1.2in,vmargin=1.2in]{geometry}
\usepackage{amsmath,amsthm,amssymb,amsfonts,verbatim}
\usepackage[mathscr]{eucal}
\usepackage{graphicx}
\usepackage[colorlinks, linkcolor=blue,  anchorcolor=blue,
citecolor=blue
]{hyperref}
\usepackage[hmargin=1.2in,vmargin=1.2in]{geometry}
\usepackage{booktabs}

\usepackage{xcolor}
\usepackage{hyperref}
\usepackage{enumitem}

\hypersetup{pdfborder = 0 0 1, colorlinks = true, linkcolor = 
blue
}
\setlist[enumerate, 1]{ label =  \textup{ (\arabic*) } }
\setlist[enumerate, 2]{ label =  \textup{ (\roman*) } }

\DeclareMathOperator{\supp}{Supp}
\DeclareMathOperator{\Div}{Div}

\DeclareMathOperator{\Gal}{Gal}
\DeclareMathOperator{\Aut}{Aut}
\DeclareMathOperator{\Diff}{Diff}
\DeclareMathOperator{\Supp}{Supp}

\DeclareMathOperator{\LC}{LC}

\renewcommand{\leq}\leqslant
\renewcommand{\geq}\geqslant
\renewcommand{\le}{\leq}
\renewcommand{\ge}{\geq}
\renewcommand{\epsilon}{\varepsilon}

\providecommand{\mcal}{\mathcal}
\providecommand{\mscr}{\mathscr}
\providecommand{\mrm}{\mathrm}
\providecommand{\msf}{\mathsf}

\providecommand{\mbb}{\mathbb}
\providecommand{\mbf}{\mathbf}
\providecommand{\bm}{\boldsymbol}

\providecommand{\Z}{{\mathbb Z}}

\providecommand{\R}{{\mathbb R}}
\providecommand{\F}{{\mathbb F}}
\providecommand{\N}{{\mathbb N}}

\providecommand{\bbp}{{\mathbb P}}

\providecommand{\igen}[1]{{\langle #1 \rangle }}

\providecommand{\abs}[1]{\ensuremath{ |#1| }}

\providecommand{\gauss}[1]{ \ensuremath {\left\lfloor #1 \right\rfloor} }
\providecommand{\ceiling}[1]{\ensuremath{ \left\lceil #1 \right\rceil } }

\providecommand{\trans}{\ensuremath{^{\msf T}}}

\newtheorem{lem}{Lemma}[section]
\newtheorem{thm}[lem]{Theorem}
\newtheorem{cor}[lem]{Corollary}

\newtheorem{prop}[lem]{Proposition}

\newtheorem{rmk}[lem]{Remark}

\providecommand{\thmref}[1]{{\hyperref[#1]{Theorem \ref{#1}}}}
\providecommand{\propref}[1]{{\hyperref[#1]{Proposition \ref{#1}}}}
\providecommand{\coref}[1]{{\hyperref[#1]{Corollary \ref{#1}}}}
\providecommand{\lemmaref}[1]{{\hyperref[#1]{Lemma \ref{#1}}}}
\providecommand{\secref}[1]{{\hyperref[#1] {Section \ref{#1}}}}

\providecommand\GK{{Giulietti--Korchm{\'a}ros\ }}
\providecommand{\QC}{quasi-cyclic\ }

\title{Multi-sequences with large linear and error linear complexity from function fields}
\author{Xubin Hu,
	Shu Liu, 
	Liming Ma
	and Chaoping Xing
	\thanks{X. Hu is with the School of Mathematical Sciences, University of Science and Technology of China, Hefei 230026, China (e-mail: xbh5168@mail.ustc.edu.cn).}
	\thanks{S. Liu is with the National Key Laboratory of Science and Technology on Communications, University of Electronic Science and Technology of China, Chengdu 611731, China.
 (E-mail: shuliu@uestc.edu.cn.)}
	\thanks{L. Ma is with the School of Mathematical Sciences, University of Science and Technology of China, Hefei 230026, China (e-mail: lmma20@ustc.edu.cn).}
	\thanks{C. Xing is with School of Electronic Information and Electrical Engineering, Shanghai Jiao Tong University, Shanghai 200240, China (email: xingcp@sjtu.edu.cn).}
	
}

\onecolumn

\begin{document}
	\maketitle
		\begin{abstract}
			The linear complexity and the error linear complexity of multi-sequences are measures for security in stream ciphers. In this manuscript, we present a general framework for constructing periodic multi-sequences via function fields. We prove that the constructed multi-sequences possess both large linear complexity and large error linear complexity. We apply this framework of constructing multi-sequences to various maximal function fields and we obtain many new multi-sequences with various lengths and dimensions. As a byproduct, we adopt this idea and produce many new quasi-cyclic algebraic geometric codes as well.
		\end{abstract}
		
	
	
\section{Introduction}
		For stream cipher, many cryptography schemes require generating long pseudorandom sequences over finite fields. For security reasons, it is favorable to use sequences with high unpredictability and high statistical randomness. The unpredictability can be assessed by various complexity measures, and the simplest one of them is the linear complexity. Meanwhile, the stability of linear complexity is also important, which can be measured by the $\epsilon$-error linear complexity introduced in \cite{DXS1991,SM1993}. As the vectorized stream cipher develops, it is necessary to adapt the two complexity measures above to suit for finitely many parallel sequences over finite fields, or multi-sequences (see \cite{NiedSurv} or below for definitions).

        Related works thrived in past three decades. It starts from theoretical researches on linear and error linear complexity of sequences. In \cite{DXS1991}, it was conjectured that there may exist trade-offs between linear complexity and the error complexity of periodic sequences. This conjecture was disproven by Niederreiter in 2003 \cite{Nied2003}. Niederreiter constructed $N$-periodic sequence with linear complexity $N$ and $\epsilon$-linear complexity $N-1$ by using cyclotomic cosets. Meidl and Niederreiter determined the expected value of both complexity measures for periodic sequences via discrete Fourier transforms in 2002 \cite{MN2002}, and for periodic multi-sequences \cite{MN2003}. A probabilistic theory was developed by Niederrieter in 2006 \cite{Nied2006}. Furthermore, Meidl et al.~developed a theory on error linear complexity of multi-sequences in 2007 \cite{MNV2007}. This theory was further refined in \cite{NV2008}, and the upper bounds of several error linear complexity measure were obtained.    These are typical examples in theoretical researches.
		
		On the other hand, people also studied linear complexity and error linear complexity of explicit sequences and multi-sequences. For example, Aly and Winterhof determined the exact value of error linear complexity of the Legendre sequences over a prime finite field \cite{AW2006}. In 2007, the linear complexity of $d$-ary Sidel'nikov sequence over a prime finite field was studied by Aly and Meidl in \cite{AM2007}. Aly et al.\ also studied the cyclotomic sequence in the aspect of error linear complexity in \cite{AMW2007}. 

Especially, the theory of function fields plays an important role for studying and constructing sequences and multi-sequences with high linear complexity and high $\epsilon$-error linear complexity. 
        Xing et al.\ constructed multi-sequences over $\F_2$ via the elliptic function fields with high linear complexity and low correlation \cite{XKD2003}, which can be refined and generalized in \cite{JMXZ2025,LZF25}. Later in 2007, Hu et al.~generalized this construction to $p$-ary sequences, and studied their periods, linear complexities, aperiodic correlations, and other properties \cite{HHF2007}.  In 2009, Xing and Ding constructed multi-sequences over $\F_{q^2}$ via cyclic automorphism subgroups of the Hermitian function fields \cite{XD2009}. They obtained three families of multi-sequences with both high linear complexity and high error linear complexity. Inspired by this work, periodic and aperiodic sequences via rational and Hermitian function fields were also constructed   \cite{ZHD2020}, and  periodic multi-sequences via a tower of Artin--Schreier extension of function fields were constructed in \cite{Tong2012}.
		
		Despite these work on constructing multi-sequences, the choice of parameters and complexity measures are still limited. It is necessary to construct multi-sequences that exhibit high complexity measures and provide flexibility on the choice of parameters such as period, dimension and etc.

        \subsection{Main Results and Techniques}

        In this paper, we provide a general framework for constructing periodic multi-sequences with high linear complexity and high $\epsilon$-error linear complexity. This framework is mainly extracted from the work of Xing and Ding \cite{XD2009}. We apply this framework to several well-known maximal function fields and obtain new multi-sequences with favorable parameters related to linear complexity and error linear complexity.

        Now we summarize the main ideas of our framework of constructing multi-sequences via function fields. We prefer to use the language of algebraic function fields. Let $F$ be an algebraic function field defined over the finite field $\F_q$.
        Each entry of the multi-sequence is defined in the form of the evaluation $f(P)$, where $P$ is a rational place of $F$ and $f$ is a function from the valuation ring of $P$. For an easier calculation of linear complexity, we want the rational places for evaluation in the same sequence to be related by a single automorphism $\sigma$ of $F/ \F_q$. Then it is natural to require most of the rational places of $F$ to be divided into long orbits under the action of cyclic group $\langle \sigma \rangle$. Equivalently, we require the existence of many rational places that split completely in the extension $F/ F^{\igen \sigma}$. Let $R$ be one of the rational place of $F$ that splits from these places of $F ^{\igen \sigma}$. The estimate of linear complexity will be converted to count zeros and poles of a certain function $\widetilde f$ derived from the linear recurrence relation that generates each sequence. We choose for each sequence the function $f$ from a Riemann--Roch space $\mcal L(G)$ with $G$ fixed by $\sigma$, and we require the support $\supp(G)$ of $G$ to avoid all evaluation places. To prevent the vanishing of $\widetilde f$, we may require $f$ has a simple pole at $R$. Therefore we choose $f$ from the set $\mcal L(G+R) \setminus \mcal L(G)$ (which is nonempty provided that $\deg (G)$ is sufficiently large). Under an additional assumption on the degree of divisors, we can prove that the obtained multi-sequence possesses both large linear complexity and large $\epsilon$-error linear complexity for relatively small positive integers $\epsilon$.

        This idea of constructing multi-sequences via function fields can be migrated to the construction of other objects that can relate to a cyclic structure. Our attempt is the construction of quasi-cyclic algebraic geometric codes. We only need to choose the Riemann--Roch space $\mcal L(G)$ as the space of functions for the algebraic geometric codes associated with the divisor $G$ and the evaluation places. Such an idea is already adopted by Tong and Ding in 2013 \cite{TD2013} and Bonini et al.~in a most recent paper \cite{BDG2026}. We apply this idea to other maximal function fields that meet our requirements stated as above and obtain many quasi-cyclic algebraic geometry codes.

        \subsection{Organization of this paper}

		This paper is organized as follows. We prepare ourselves with necessary tools from the theory of algebraic function fields in \secref{sec:prelim}, and we simultaneously recall the concept of linear complexity and error linear complexity of multi-sequences. In \secref{sec:general_framework}, we propose the general framework for constructing multi-sequences from function fields, and provide the lower bounds of $\epsilon$-error linear complexity for relatively small integer $\epsilon$.  In  \secref{sec:instances}, we apply this framework to cyclotomic function fields, maximal function fields of second largest genus, Suzuki function fields and \GK function fields to provide new multi-sequences with large linear complexity and $\epsilon$-error linear complexity. Finally, in \secref{sec:QC_codes} we exploit the similarity of periodic multi-sequences and quasi-cyclic codes to produce many new  \QC algebraic geometric codes.

	\section{Preliminaries}
	
	\label{sec:prelim}
	In this section, we collect necessary knowledge for the construction of multi-sequences from function fields. We present preliminaries and basic results of function fields, extension theory of function fields and multi-sequences with their linear complexity and $\epsilon$-error linear complexity. Details for function fields can be found in \cite{GTM254, NXBook}. The reader may refer to \cite{NiedSurv} for more information on linear complexity and $\epsilon$-error linear complexity of multi-sequences.
	\subsection{Algebraic function fields}
	
	Let $q$ be a power of a nonnegative prime integer $p$, and $\F_q$ be the finite field with exactly $q$ elements. An algebraic function field $F$ over $\F_q$ is a finite algebraic extension of the rational function field over $\F_q$. We also assume that $\F_q$ is algebraically closed in $F$, and we usually write $F / \F_q$.  The set of all places of $F$ is denoted by $\mbb P_F$. Let $P$ be a place of $F$, and $\mcal O_P$ be its corresponding valuation ring. For $z \in \mcal O_P$, $z(P)$ is defined to be the residue class of $z$ modulo $P$ in $\kappa_P:= \mcal O_P / P$. The degree $\deg(P)$ of the place $P$ is the extension degree $[\kappa_P : \F_q]$. The places of $F$ of degree $1$ are called rational. They form a subset $\mbb P_F^{(1)}$ of $\mbb P_F$.
	
	A divisor of $F$ is a formal sum $D = \sum _{P \in \mbb P_F} n_P P$ with integer coefficients $n_P \in \Z$, where $n_P \neq 0$ for at most finitely many places $P$. The degree of $G$ is defined to be $\deg (G) := \sum n_P \deg (P)$. The divisor $D$ is called effective if all coefficients $n_P \geq 0$, and we write $D \geq 0$. The support of the  divisor $D$ is the set
	\[
	\Supp (D) = \{P \in \mbb P_F \colon n_P \neq  0\}.
	\] The divisors of $F$ form an abelian group $\Div (F)$. For each place $P$, let $v_P$ be the corresponding normalized discrete valuation. For a nonzero element $z \in F$, a place $P$ is a zero of $z$ if $z \in P$, and it is a pole of $z$ if $z^{-1} \in P$. The zero divisor of $z$ is the formal sum
	\[
	(z)_0 = \sum _{P\colon z \in \mcal O_P} v_P (z) P,
	\]
	and the pole divisor is $(z)_\infty := (1/z)_0$. The principal divisor of $z$ is just
	\[
	(z) :=\sum _{P \in \mathbb{P}_F} v_P (z) P= (z)_0 - (z)_\infty.
	\]

	For a divisor $G$ of $F$, the Riemann--Roch space associated with $G$ is defined as
	\begin{align*}
		\mcal L(G)\,\, {:=}\,\,  \{ z \in F^\ast \colon (z) +G \geq 0 \}\cup \{0\}.
	\end{align*}
	It is a finite-dimensional vector space over $\F_q$ of dimension $\dim (G)$. The celebrated Riemann--Roch Theorem shows that the integer
	$
	\deg (G) - \dim (G) + 1
	$
	is bounded from above for all divisors $G$. The lowest upper bound $g(F)$ is called the genus of $F$. The genus $g(F)$ is nonnegative, and the equation
	\[
	\dim (G) = \deg (G) -g (F)  + 1
	\]
	holds true for any divisor $G$ with $ \deg(G)\geq 2g(F) - 1$ \cite[Theorem 1.5.17]{GTM254}.
	
	For a function field $F/\F_q$, an $\F_q$-automorphism of $F$ is an automorphism $F \to F$ that fixes all elements in $\F_q$. They form a group $\Aut (F/ \F_q)$, which is called the automorphism group of $F$ over $\F_q$. For any $\sigma \in \Aut (F / \F_q)$ and a place $P$ of $F$, the set $\sigma (P) = \{\sigma (z) \colon z  \in P\}$ is also a place of $F$ with the same degree as $P$. This defines an action on $\mbb P_F$ by the automorphism group $\Aut (F/ \F_q)$. This action extends $\Z$-linearly to $\Div (F)$.  From \cite[Lemma 1]{JMXZ2025}, we can see that for each divisor $D$, $\sigma (\mcal L(D)) = \mcal L(\sigma (D))$. In particular if $\sigma (P) =P$, then  $ \sigma (\mcal L(r P)) = \mcal L(rP) $ for any $r \in \N$.

    \subsection{Extension theory of function fields}
	Let $E / \F_q$ be a finite Galois algebraic extension of $F$. The Hurwitz genus formula \cite[Theorem 3.4.13]{GTM254} states that
    \[
	2 g(E) - 2 = [E :F](2 g (F) - 2) + \deg (\Diff (E /F)),
	\]
    where $ \Diff (E/F)$ is the different divisor of $E / F $. For $P \in \bbp _F$ and $ Q \in \bbp_E$ such that $ Q | P$, let $ d(Q|P), e(Q|P), f(Q|P)$ be the different exponent, ramification index and relative degree of $Q | P$, respectively. Then $ \Diff (E / F)$ is the divisor $ \sum _{Q \in \bbp  _E} d(Q|Q \cap F) Q $. If $p \nmid e (Q|P)$, then $ d(Q|P)= e(Q|P) - 1$ by the Dedekind different theorem \cite[Theorem 3.5.1]{GTM254}.

    Let $P \in \mbb P_F$ and $Q \in \mbb P_E$ with $Q |P$.  The decomposition group is $G_Z(Q | P) = \{\sigma \in \Gal (E/F) \colon \sigma (Q) = Q \} $. The subfield of $E$ fixed by $G_Z(Q|P)$ is called the decomposition field of $Q|P$. There exists a surjective homomorphism $\varphi \colon  G_Z (Q|P) \to \Gal (\kappa_Q / \kappa_P), \sigma \mapsto \overline \sigma$, where $\overline \sigma ( z(Q) ) := (\sigma z) (Q) $ for each $z \in \mcal O_Q$. Now assume $Q|P$ is unramified. By \cite[Proposition 1.4.6]{NXBook}, the homomorphism $\varphi$ is an isomorphism.
    Since $\kappa_Q / \kappa _P$ is an extension of finite fields, the group $\Gal (\kappa_Q / \kappa_P)$ is cyclic with a generator. The preimage of this generator under the isomorphism $\varphi$ is called the Frobeinus symbol of $Q$ over $P$. By \cite[Theorem 1.4.11]{NXBook}, this symbol is independent of $Q$ when $E / F$ is an abelian extension. In this case this symbol is called the Artin symbol $\left[ \frac {E /F} P \right]$ of $P$ in $E/F$. The following lemma given by {\cite[Proposition 1.4.12]{NXBook}} is useful to determine whether an unramified place splits completely or not.
	\begin{lem}
		Let $E/F$ be a finite abelian extension of function fields, and let $M$ be an intermediate field of $E/F$. Let $P$ be an place of $F$ that is unramified in $E/F$. Then $P$ splits completely in $M/F$ if and only if the Artin symbol $\left[\frac {E/F} P\right] \in \Gal (E/M)$.  \label{lem:crit_tot_split}
	\end{lem}

\subsection{Multi-sequences and their linear complexity}
	Let $\F_q$ be the finite field of $q$ elements. A sequence $\bm s$ from $\F_q$ is simply $(s_0, s_1, s_2, \dots )$ with $s_j\in \F_q$ for any $j\in \mathbb{N}$.  If there exist a positive integer $m$ such that $s_{j+m} = s_j$ for all $j \geq 0$, then this sequence is called periodic. We identify the periodic sequence with a sequence $\bm s = (s_0, s_1, \dots, s_{m-1}) \in \F_q^m$ of finite length $m$. From now on, we focus on periodic sequences over $\F_q$. A finite subset $\bm S \subseteq \F_q^m \setminus \{\mbf 0\}$ is called a multi-sequence over $\F_q$ of length $m$ and dimension $\ell = |\bm S|$. Let $\bm S = \{\bm s_1, \bm s_2, \dots, \bm s_\ell \}$, and we denote the terms of the $k$-th sequence $\bm s_k$ by $s_{k, 0}, s_{k,1},\dots, s_{k, m-1}$. The  linear complexity $\LC(\bm S)$ of $\bm S$ is defined to be the least nonnegative integer $L$ for which there exists $\lambda_1, \lambda_2, \dots, \lambda_{L} \in \F_q$ such that
    \[
        s_{k, j} = \lambda_{L}s_{k, j-1} + \lambda_{L-1} s_{k, j-2}  +\dots + \lambda_1 s_{k, j - L}
    \]
    for all $1 \leq k \leq \ell$ and $0\leq j \leq m-1$, where the indexes involving $j$ are always taken to be the least nonnegative residue modulo $m$. In other words, $\LC(\bm S)$ is the least order of a linear recurrence over $\F_q$ that simultaneously generates all sequences $\bm s_k$ for $1\le k\le \ell$. The polynomial $T^L - \sum_{i = 1}^{L} \lambda_{i}T^{i-1} \in \F_q [T] $ is called the generating polynomial of the multi-sequence $\bm S$. It is evident from the definition that if $\bm S' \subseteq \bm S$ is a subset, then $\LC (\bm S') \leq \LC(\bm S)$.

    We identify the multi-sequence $\bm S$ as the $\ell \times m$ matrix $  [ \bm s_1 \trans, \bm s_2 \trans, \dots, \bm s_\ell \trans]\trans $, where $(-)\trans$ stands for the transpose of a vector or matrix. For an integer $0 \leq \epsilon \leq \ell m$, the  $\epsilon$-error linear complexity is defined by
    \[
        \LC _\epsilon(\bm S) = \min _{\bm T} \LC (\bm T),
    \]
    where $\bm T$ runs through all $\ell \times m$ matrix obtained by changing at most $\epsilon$ terms in $\bm S$. A quick observation: if $0 \leq \epsilon \leq \epsilon' \leq \ell m$, then $\LC_\epsilon (\bm S) \geq \LC_{\epsilon'} (\bm S) $.

\section{A general framework of multi-sequences via function fields}
	
	\label{sec:general_framework}
	In this section, we propose a general framework of explicit construction of multi-sequences with large joint linear complexity and large joint $\epsilon$-error linear complexity via automorphisms of function fields.
	
	Let $F$ be a global function field over finite field $\F_q$ with genus $g$. Let $\sigma $ be an $\F_q$-automorphism of $F$ of order $m$. Let $\ell$ be a positive integer. Assume that there exist $\ell + 1$ rational places of the fixed subfield $F^{\igen {\sigma}}$ which split completely in the extension $F  / F^{\igen{\sigma}}$. Pick one of such place from $F^{\igen{\sigma}}$, and label all rational places above it as $R, \sigma(R), \dots, \sigma^{m-1} (R)$. Label other mentioned $\ell m$ rational places of $F$ as $P_k, \sigma (P_k), \dots, \sigma ^{m-1} (P_k)$ for $k = 1,2, \dots, \ell$. Let $G$ be a divisor of $F$ fixed by $\sigma$, i.e., $\sigma(G)=G$, $\deg (G) \geq 2g  - 1$ and
	\[
		\supp (G) \cap \{ R, \sigma(R), \dots, \sigma ^{m-1}(R), P_1, \dots, \sigma ^{m-1}(P_1), \dots,P_\ell, \dots,  \sigma ^{m-1} (P_\ell)\} = \varnothing.
	\] By the Riemann--Roch Theorem \cite[Theorem 1.5.17]{GTM254}, we obtain that $\dim_{\F_q} \mcal L(G) = \deg (G) - g + 1$ and $\dim_{\F_q} \mcal L(G+R) = \deg (G) - g+2$. It follows that there exists a function $0 \neq f \in \mcal L(G+R) \setminus \mcal L(G)$. Hence, we have $(f)_\infty = R + G'$, where $G'$ is an effective divisor with $G \geq G'$.  For $k = 1,2,\dots, \ell$, consider the sequences
	\[
	\bm {s}_k (f) := (f (\sigma^{j - 1} (P_k))_{j = 1}^{m}.
	\]
	It is clear that $\bm S = \{\bm s _k \}_{k = 1}^\ell$ is a periodic multi-sequence of dimension $\ell$ and period $m$. Furthermore, we shall prove that the linear complexity of this multi-sequence $\bm S$ is $\operatorname{LC} (\bm S) = m$ under some additional assumptions.
	
	\begin{thm}
		Let $\bm S$ be the multi-sequence constructed above. If $(\ell -1)m \geq 2g$ and $2g - 1 \leq \deg (G) \leq (\ell - 1)m - 1 $, then the linear complexity of the periodic multi-sequence $\bm S$ is $m$.
		\label{thm:mult_seq_framework}
	\end{thm}
	
	\begin{proof}
		Suppose that the linear complexity of the periodic multi-sequence $\bm S$ was $\mrm {LC}(\bm S) = \beta < m$. By definition, there exist $\lambda_0, \lambda_1, \dots, \lambda_{\beta -1}, \lambda_\beta = 1 \in \F_q$ such that
		\[
		\sum_{i = 0}^\beta \lambda_i f(\sigma^{i + s} (P_k)) = 0
		\]
		for all integers $s \geq 0$ and $k = 1,2,\dots, \ell$. By \cite[Lemma 3.5.2]{GTM254}, we obtain that $$ f(\sigma^{i + s} (P_k))=f(\sigma^i(\sigma^s(P_k))) = \Bigl(\sigma ^i(\sigma^{-i}(f)\Bigr) \Bigl(\sigma ^i (\sigma ^s (P_k)) \Bigr) = \Bigl(\sigma^{-i}(f) \Bigr) \Bigl(\sigma ^s(P_k) \Bigr)  $$
        for all integers $s \geq 0$ and $k = 1,2,\dots, \ell$. It follows that
		\begin{equation}
		\left(\sum_{i = 0}^\beta \lambda_i \sigma ^{-i }(f)  \right) (\sigma ^s (P_k)) =  0\label{eqn:tilde_f_zeros}
		\end{equation}
		holds for all integers $s \geq 0$ and $k = 1,2,\dots, \ell$. From the choice of $f \in \mcal L(G+R) \setminus \mcal L(G)$, we have $v_R(f) = -1$. Since $\sigma^{-i}(R) $ are pairwise distinct for $0\leq i\leq \beta$, the place $\sigma ^{-\beta} (R)$ is a pole of $\lambda_\beta \sigma ^{-\beta }(f)$ and is not a pole of $\sigma ^{-j }(f)$ for any $0\le j< \beta$ by \cite[Lemma 3.5.2]{GTM254}.
		Thus, the function $\tilde f := \sum _{i = 0}^\beta \lambda_i \sigma ^{-i}(f)$ is not zero, since $\tilde f$ must have a pole at $\sigma ^{-\beta} (R)$ by the Strict Triangle Inequality \cite[Lemma 1.1.11]{GTM254}. Therefore, $\tilde f$ is a nonzero element in the Riemann--Roch space $\mcal L(G + \sum_{i=0}^\beta \sigma^{-i}  (R))$. From Equation \eqref{eqn:tilde_f_zeros}, we see that $\sigma ^s (P_k)$ are zeros of $\tilde f$ for all $0\le s\le m-1$ and $1\le k\le \ell$. It follows that
		\[
		0 \neq \tilde f \in \mcal L\left(
		G+ \sum_{i=0}^\beta \sigma^{-i }(R) - \sum_{k = 1}^\ell \sum_{s = 0}^{m-1} \sigma ^{s}(P_k)
		\right).
		\]
		Therefore, the divisor $	G+ \sum_{i=0}^\beta \sigma^{-i }(R) - \sum_{k = 1}^\ell \sum_{s = 0}^{m-1} \sigma ^{s}(P_k)$ has degree at least $0$, i.e.,
		\[
		\deg (G)  - \ell m  + \beta+1 \geq 0.
		\]
		From the assumption on $\deg (G)$, it is clear that $\beta \geq \ell m  -\deg (G)  -1 \geq m$, which is a contradiction. Thus, we obtain that $\LC(\bm S) = m$.
	\end{proof}
	
	\begin{cor}\label{cor:sub_multi_seq_same_LC}
		Let $\delta$ be a positive integer satisfying $\delta \geq 1+(\deg(G)+1)/m$.
		Any $\delta$ sequences in $\bm S$ form a multi-sequence of linear complexity $m$ as well.
	\end{cor}
	\begin{proof}
		Let $\bm S_0\subseteq \bm S$ be a subset of $\bm S$ with $\delta$ sequences. Without loss of generality, we can assume that $\bm S_0 = \{\bm s_i(f)\}_{i = 1}^\delta$.
		By a similar argument given in  \thmref{thm:mult_seq_framework}, we can prove this corollary. If the linear complexity of $\bm S_0$ is $\beta = \operatorname{LC}(\bm S_0) < m$, then there exists a nonzero function $\tilde f = \sum _{i = 0}^\beta \lambda_i \sigma ^{-i}(f)$ in the Riemann--Roch space
		\[\mcal L\left(
		G+ \sum_{i=0}^\beta \sigma^{-i }(R) - \sum_{k = 1}^\delta \sum_{s = 0}^{m-1} \sigma ^{s}(P_k)
		\right).\]
		Therefore, we can obtain an inequality
		\[
		\deg(G) + \beta + 1 -\delta m \geq 0,
		\]
		which implies $\beta \geq \delta m - \deg(G) - 1 \geq m$. Thus, we derive a contradiction.
	\end{proof}

	If there are not many errors occurred in the multi-sequences, then we can obtain the following results on the $\epsilon$-error linear complexity of the multi-sequence $\bm S$.
	\begin{thm}
		Let $\bm S$ be the multi-sequence constructed at the beginning of this section. Let $\delta _ 0 = 1  + \ceiling{(\deg (G) + 1) / m}$, where $\ceiling {x}$ is the least integer that is not less than $x \in \R$. When $2 g - 1 \leq \deg (G) \leq (\ell - 2) m - \ell + \delta_0 - 3$, we have the following conclusions.
		\begin{enumerate}
			\item The $\epsilon$-error linear complexity of $\bm S$ is $m$ if $\epsilon \leq \ell - \delta _0$.
			\item The $\epsilon$-error linear complexity of $\bm S$ is at least $\chi_1:= -1 + (\ell m - \deg (G)) / (\ell - \delta _ 0 + 2)$ if $\epsilon \leq \ell - \delta _0 + 1$.
			\item If $\epsilon \leq  \ell - \delta _ 0 + 2$, then  the $\epsilon$-error linear complexity of $\bm S$ satisfies
			\[
				\LC_{\epsilon}(\bm S) \geq -1+ \min \left\{ \chi_2:=
				 \frac {(\ell - 1)m - \deg (G)} {\ell - \delta_0 + 1}, \chi_3:= \frac {\ell m -\deg (G)} {\ell - \delta _0 + 3}
				\right\}.
			\]
		\end{enumerate} \label{thm:general_eps_err}
	\end{thm}
	
	\begin{proof}
		\begin{enumerate}
			\item Let $\bm S_1$ be any multi-sequence obtained by altering at most $\epsilon \leq \ell - \delta_0$ entries in $\bm S$. Then $\abs  { \bm S \cap \bm S_1} \geq \ell - \epsilon \geq \delta_0$. By \coref{cor:sub_multi_seq_same_LC}, the linear complexity of $\bm S\cap \bm S_1$ is $\operatorname{LC}(\bm S \cap \bm S_1) = m$. It is easy to see that $m \geq \LC (\bm S_1) \geq \LC(\bm S \cap \bm S_1) = m$. Thus, $\LC(\bm S_1) = m$. When $\bm S_1$ runs through all such multi-sequences, it is clear that $\LC _{\epsilon} (\bm S) = m$.

			\item Let $\bm S_2$ be any multi-sequence obtained by altering at most $\epsilon\leq \ell - \delta_0 + 1$ terms in $\bm S$. Then we have $ \abs {\bm S \cap \bm S_2}\geq \ell - \epsilon \geq \delta_0 - 1$. If $\abs {\bm S \cap \bm S_2} \geq \delta_0$, then $\LC (\bm S_2) = m$ from item (1). Now assume $\abs {\bm S \cap \bm S_2} = \delta_0 - 1$, i.e.\ $\ell - \delta_0 + 1$ sequences from $\bm S$ have changed. It is clear that each changed sequence has exactly one altered entry. We may assume that $\bm S \cap \bm S_2 = \{ \bm s_j (f) \}_{j = 1}^{\delta_0 - 1}$, and the unique entry $f(\sigma ^{w_k }(P_k))$ with $0 \leq w_k \leq m-1$ is altered for each $\delta_0 \leq  k \leq \ell$. Assume that $\LC (\bm S_2) = \beta$ and the polynomial $\sum_{j = 0}^\beta \lambda_j T^j$ generates $\bm S_2$ with $\lambda _\beta = 1$. Then we have
			\[
				\sum_{j = 0}^\beta \lambda_j f (\sigma ^{j + s } (P_k)) = 0
			\]
			for all $k = 1,2, \dots, \delta_0 - 1$ and $s \geq 0$; moreover,
			\[
				\sum_{j = 0}^{\beta} \lambda _j f(\sigma ^{j + s} (P_k)) = 0
			\]
			for all $k = \delta_0, \delta_0+1, \dots, \ell$ and  $w_k + 1 \leq s \leq m + w_k -1 - \beta$. Thus, the function $\sum _{j = 0}^\beta \lambda _j \sigma ^{-j} (f)$ belongs to the nonzero Riemann--Roch space
			\[
				\mcal L \left(
				G + \sum _{j = 0}^\beta \sigma ^{-j }(R) - \sum _{k = 1}^{\delta_0 - 1} \sum _{s = 0}^{m-1} \sigma ^{s} (P_k) - \sum_{k = \delta_0}^{\ell } \sum _{s = w_k + 1} ^{m + w_k - \beta - 1} \sigma ^{s}(P_k)
				\right).
			\]
			Therefore, the associated divisor of this space has a nonnegative degree, which yields
			\[
				\deg (G) + \beta + 1 \geq (\delta_0 - 1) m + (\ell - \delta _ 0 + 1) (m - \beta - 1).
			\]
			Similarly, this item follows from
			\[
				\LC (\bm S_2) = \beta \geq - 1 + \frac {\ell m -\deg (G)} {\ell - \delta _0+ 2}.
			\]

			\item Let $\bm S_3$ be any multi-sequence obtained by altering exactly $\epsilon =  \ell - \delta_0 + 2$ entries in $\bm S$. Then we have $\abs { \bm S \cap \bm S_3} \geq \delta _ 0 - 2$. We shall discuss this item by cases.
			\begin{enumerate}
				\item If $\abs {\bm S \cap \bm S_3} \geq \delta_0$, then $\LC (\bm S_3) = m$ from item (1).
				\item If $\abs {\bm S \cap \bm S_3} = \delta_0 - 1$, i.e., $\ell -\delta_0 + 1$ sequences in $\bm S$ get changed, then there exist $\ell - \delta_0 $ sequences with exactly one altered entry and  one sequence with two altered entries by the pigeon-hole principle. Assume that $\bm S \cap \bm S_3 = \{  \bm s_k (f)\}_{k = 1}^{\delta_0 - 1}$ and   the only altered entry of $\bm s_k (f)$ is $f (\sigma ^{w_k} (P_k))$ with some $0 \leq w_k \leq m-1 $ for each $k  = \delta_0, \dots, \ell - 1$. Let $\LC (\bm S_3) = \beta$ and the polynomial $\sum_{j = 0}^\beta \lambda _j T^j$ generate $\bm S_3$ with $\lambda_\beta = 1$. By a similar argument, there exists a nonzero element in the Riemann-Roch space
				\[
					\mcal L \left(
					G + \sum _{j = 0}^\beta \sigma ^{-j }(R) - \sum _{k = 1}^{\delta_0 - 1} \sum _{s = 0}^{m-1} \sigma ^{s} (P_k) - \sum_{k = \delta_0}^{\ell -1} \sum _{s = w_k + 1} ^{m + w_k - \beta - 1} \sigma ^{s}(P_k)
					\right).
				\]
				Hence, such a divisor associated to the above space has a nonnegative degree, i.e.,
				\[
					\deg (G) + \beta + 1 \geq (\delta_0 - 1)m + (\ell - \delta_0) (m - \beta - 1).
				\]
				This derives
				\[
					\LC (\bm S_3)=\beta \geq -1+ \frac {(\ell - 1) m -\deg (G)} {\ell - \delta_0 + 1}.
				\]
			
				\item If $\abs {\bm S \cap \bm S_3} = \delta _0 - 2$, i.e., $\ell - \delta_0 + 2$ sequences from $\bm S$ get changed, then each of $\ell - \delta_0 + 2$ sequences has exactly one altered entry. Let $\beta = \LC (\bm S_3)$. By an analogous argument as in item (2), we obtain that
                \[
                    \deg (G) + \beta + 1 \geq (\delta_0 - 2) m + (\ell - \delta_0 + 2) (m - \beta - 1).
                \]
               It follows that
				\[
					  \LC (\bm S_3)=\beta \geq -1 + \frac {\ell m - \deg (G)} {\ell - \delta _0 +3}.
				\]
			\end{enumerate}
			Moreover, if $\epsilon \leq \ell - \delta_0 + 1$ is allowed, then we can obtain
			\[
				\LC _\epsilon (\bm S) \geq -1 +  \min \left\{
				\frac {(\ell - 1)m - \deg (G)} {\ell - \delta_0 + 1}, \frac {\ell m -\deg (G)} {\ell - \delta _0 + 3}
				\right\}
			\]
            after comparing with the bound in item (2).
		\end{enumerate}
	\end{proof}

	\section{Multi-sequences from various function fields}
	
	\label{sec:instances}
	
	In this section, we explicitly construct many multi-sequences with large  linear complexity and large  error linear complexity from various function fields, such as subfields of cyclotomic function fields, the maximal function field with the second largest genus, the Suzuki function field, and the \GK function field.

	\subsection{Cyclotomic function fields}
	
		In this subsection, we apply the framework of constructing multi-sequences in \secref{sec:general_framework} to subfields of cyclotomic function fields.
	
	We first recall the definition of cyclotomic function fields. Let $\F_{q^2}$ be the finite field of $q^2$ elements and let $K = \F_{q^2} (x)$ be the rational function field over $\F_{q^2}$ of one variable. Fix an algebraic closure $\overline K$ of $K$. Consider an endomorphism on $\overline K$ given by
	\[
	     \phi (z)= z^x:= z^{q^2} + xz, \quad \text{for all }z \in \overline {K}.
	\]
	Then $\overline K$ can be viewed as an $\F_{q^2}[x]$-module with the action $ z^{f} := (f(\phi))(z)$ for all $f \in \F_{q^2} [x] $. For any $M \in\F_{q^2}[x]$, let $\Lambda_M$ be the $M$-torsion submodule of $\overline K$, that is
	\[
	\Lambda _M := \{  z \in \overline K\colon z^M = 0\}.
	\]
	The subfield of $\overline K$ that is generated by $K$ and $\Lambda_M$ is called the cyclotomic function field of modulus $M$ over $\F_{q^2}(x)$, and it is denoted by $K(\Lambda _M)$.
	In this paper, we focus on the case when the modulus is irreducible over $\F_{q^2}$. The following results can be summarized from \cite{Hay1974} and \cite[Theorem 3.2.6-3.2.9]{NXBook}.
	
	\begin{lem} \label{prop:Cyclt_Fun_Fld_Irr_Mod}
		Let $p(x)$ be an irreducible polynomial over $\F_{q^2}$ of degree $d$. The following results hold true for the cyclotomic function field $F = K (\Lambda _{p(x)} )$.
		\begin{enumerate}
			\item $\Lambda_{p(x)}$ is a cyclic $\F_{q^2}[x]$-module with a generator $\lambda$ and $F  =K(\lambda)$.
			\item $F$ is a finite abelian extension of $K$ with Galois group $\Gal(F / K) \cong (\F_{q^2} [x] /(p(x)))^\ast$, and the automorphism  $\sigma_f$ of $F$ corresponding to $\overline f \in \F_{q^2}[x]/(p(x))$ is determined by $\sigma_f(\lambda) = \lambda^f$.
			\item There exists only one place $P = (p(x))_0$ of $K$ that is totally ramified in $F/ K$.
			\item The infinite place $\infty $ of $K$ splits into $(q^{2d} - 1)/(q^2-1)$ places of $F$, and for each place $P_\infty $ of $F$ that lies above $\infty$, the ramification index of $P_\infty | \infty$ is $ q^2 - 1$.
			\item Other places of $K$ are unramified in $F/K$.
			\item For $r(x) \in \F_{q^2}[x]$ that is irreducible and not divisible by $p(x)$, the Artin symbol associated with the place $R = (r(x))_0^K$ is given by
			\[
				\left[\frac {F/K} {R}\right] (\lambda) = \lambda^{r(x)}.
			\]
		\end{enumerate}
	\end{lem}
	
		\begin{prop} \label{prop:Spl_Cycl_Fun_Fld}
		Let $F$ be the cyclotomic function field of modulus $p(x)$ over $K = \F_{q^2}(x)$. Let $A$ be the subgroup of $\Gal (F / K )\cong (\F_{q^2}[x] /(p(x)))^\ast $ generated by subgroups $(\F_q[x] / (p(x)) )^\ast $ and $\F_{q^2}^\ast$. Let $M = F^A$ be its fixed subfield. There are at least $q+1$ rational places of $K$ that split completely in $M/K$, and the genus of $M$ is
		\[
		g(M) = \frac {d-2}2 \left(\frac {q^d + 1} {q+1} - 1\right).
		\]
	\end{prop}
	
	\begin{proof}
		Let $Z$ be the decompostion field of the infinity place $\infty$. By item (4) of \lemmaref{prop:Cyclt_Fun_Fld_Irr_Mod}, the place $\infty$ splits completely in $Z / K$, and $\Gal (F / Z) \cong \F_{q^2}^\ast$ is the unique subgroup of $\Gal (F/ K )$ of order $q^2 - 1$. Since $\Gal (F/ Z) \cong \F_{q^2}^\ast \leq A$, the fixed subfield $M= F^A$ is an intermediate field of $Z /K$. Therefore, $\infty$ splits completely in $M /K$.

        For each $\alpha \in \F_q$, the place $P_\alpha = (x - \alpha)_0^K$ has Artin symbol given by $\lambda \mapsto \lambda ^{x - \alpha}$. It is easy to confirm that $x - \alpha \mod (p(x))$ lies in $(\F_q[x] / (p(x)))^\ast $. By the item (3) of \lemmaref{prop:Cyclt_Fun_Fld_Irr_Mod}, the Artin symbol of $P_\alpha$ in $F/K$ is $\left[\frac {F/K} {P_\alpha}\right] \in (\F_q[x] / (p(x)))^\ast \subseteq A$. By \lemmaref{lem:crit_tot_split}, the place $P_\alpha$ splits completely in $M /K$ for each $\alpha \in \F_q$. Therefore, there are at least $q+1$ rational places of $K$ that split completely in $M/K$.
		
		Let $\F_{q^2}^\ast = \igen \zeta$, then $\F_q^\ast$ can be viewed as the subgroup $\igen {\zeta ^{q+1}}$. Then a nonzero polynomial representative of an element in $\F_q[x] / (p(x))$ is actually $\sum_{j = 0}^{d-1} a_j x^j$, where $a_j \in \igen {\zeta^{q+1}} \cup \{ 0\} $ but not all $a_j$ are $0$. Thus, the group $A$ is in bijective correspondence with the set
		\[
		\left\{ \zeta ^{k} \sum_{j = 0}^{d-1} a_j x^j \colon a_0, \dots, a_{d-1} \in \F_q, k = 0,1,2,\dots, q \right\} \setminus \{0\}.
		\]
		Thus, $\abs A = (q^d - 1) (q+1)$ and the degree of extension $ [M:K] =[F:K] / \abs A = (q^d + 1) / (q+1)$.  By  \lemmaref{prop:Cyclt_Fun_Fld_Irr_Mod}, the place $P$ is the unique ramified place of $K$ in the extension $M/K$. Moreover, it is totally ramified and tamely ramified. By applying Hurwitz genus formula to $M/K$, we have
		\[
		2 g (M) - 2 = \frac {q^d + 1}{q+1} (0 - 2) + d \cdot \left(\frac {q^d + 1}{q+1} - 1\right).
		\]
		It follows that
		\[
		g(M) = \frac {d-2}2 \left(\frac {q^d + 1} {q+1} - 1\right). \qedhere
		\]
	\end{proof}

    \begin{rmk}
        This proposition is similarly stated and briefly explained in \cite[Section II-D]{LXY2017}. For the completeness of this paper, we present the detailed proof here.
    \end{rmk}

	\begin{prop}
		Let $q \geq 5$ be a power of prime, $3 \leq d \leq q$ be an odd integer and $\psi(d) = (q^d + 1)/ (q+1)$.
        Then there exists a multi-sequence over $\F_{q^2}$ of dimension $q$, period $\psi(d)$ and linear complexity $\psi (d)$.  \label{prop:mult_seq_construct_cycl}
	\end{prop}
	
	\begin{proof}
		Let $\sigma $ be a generator of $\Gal (M/K)$ which is an $\F_{q^2}$-automorphism of order $\psi(d)$. By \propref{prop:Spl_Cycl_Fun_Fld}, there are $q+1$ rational places of $K$ splitting completely in $M$. Let $Q$ be the unique place of $M$ that lies above the zero $P$ of $p(x)$ in $K$. By \lemmaref{prop:Cyclt_Fun_Fld_Irr_Mod}, $Q|P$ is totally ramified and hence $\sigma (Q) = Q$. Let $t = \gauss {(2g(M) - 2) / d} + 1$ and $G = tQ$. Then the divisor $G$ is fixed by $\Gal(M/K)$ and $2 g(M) - 1 \leq \deg (G) \leq (q-1)\psi (d)-1 $. This proposition follows from \thmref{thm:mult_seq_framework}.
	\end{proof}
	
	Moreover if $\epsilon$ is small, then we can deduce the following results on the $\epsilon$-error linear complexity of this multi-sequence.
	
	\begin{prop}
		Let $\bm S$ be the multi-sequence constructed in \propref{prop:mult_seq_construct_cycl}. Let $\psi(d) = (q^d + 1)/ (q+1)$.
        If $d \mid (d - 2) \psi (d)$, then we have the following results.
		\begin{enumerate}
			\item The $\epsilon$-error linear complexity of $\bm S$ is $\psi (d)$ for $\epsilon \leq q-d$.
			\item The $\epsilon$-error linear complexity of $\bm S$ is at least $\psi (d) - 1$ for $\epsilon \leq q - d +1$.
			\item The $\epsilon$-error linear complexity of $\bm S$ is at least $\psi (d) - 1 - {\psi(d)} /(q - d+3) $ for $\epsilon \leq q - d + 2$.
		\end{enumerate}
	\end{prop}
	
	\begin{proof}
		In the setting of \thmref{thm:general_eps_err}, $\ell  = q$, $m = \psi (d)$, and $\deg (G) = td = (d-2)\psi (d)$. It follows that $\delta_0 = 1 + \ceiling{(1 + \deg (G))/m } = d$.
        \begin{enumerate}
            \item If $\epsilon \leq  \ell - \delta_0=q - d$, then the  $\epsilon$-error linear complexity is $\operatorname{LC}_{\epsilon} (\bm S) = \psi(d)$.
            \item If $\epsilon \leq  \ell - \delta_0+1= q - d + 1$, then we have
            \[
			\LC_\epsilon (\bm S) \geq \chi_1 = -1 + \frac {q \psi (d) - (d - 2)\psi (d)}{q - d+2} = \psi(d ) - 1.
		     \]
            \item If $\epsilon \leq  \ell - \delta_0+2=  q- d + 2$, then it is easy to verify that
            \[
            \chi_2 = \psi (d) \geq \psi(d) - \frac {\psi(d)} {q - d + 3} = \chi_3.
            \]
            Thus, we have
            \[
			\LC _\epsilon (\bm S) \geq -1 + \min \{\chi_2, \chi_3\} = \psi(d) - \frac {\psi (d)} {q - d + 3} -1. \qedhere
		      \]
        \end{enumerate}
	\end{proof}
	
	\begin{prop}\label{prop4.6}
		Let $\bm S$ be the multi-sequence constructed in \propref{prop:mult_seq_construct_cycl}.
        Let $\psi(d) = (q^d + 1)/ (q+1)$.
        If $d \nmid (d - 2) \psi (d)$, then we have the following results.
		\begin{enumerate}
			\item The $\epsilon$-error linear complexity of $\bm S$ is $\psi (d)$ for $\epsilon \leq q-d+1$.
			\item The $\epsilon$-error linear complexity of $\bm S$ is at least $\psi (d) - 1 - (\psi(d)+1)/(q - d+3)$ for $\epsilon \leq q - d +2$.
			\item The $\epsilon$-error linear complexity of $\bm S$ is at least $\psi (d) - 1 - (2\psi(d) - 1) /(q - d+4) $ for $\epsilon \leq q - d + 3$.
		\end{enumerate}
	\end{prop}
	
	\begin{proof}
		In the setting of  \thmref{thm:general_eps_err}, $\ell = q$, $m = \psi(d)$, and $\deg (G) = td$. By \propref{prop:Spl_Cycl_Fun_Fld}, we have $2 g(M) -2 = (d - 2) (\psi(d) - 1) - 2 = (d-2) \psi (d) - d$.  By the assumption, we have
			$(d - 2)\psi (d) - d+1 \leq td \leq (d- 2) \psi (d) - 1.$
		It follows that $$\delta_0 = 1 + \ceiling{(td + 1) / \psi(d)} = d-1.$$
		\begin{enumerate}
			\item If $\epsilon \leq q - \delta_0=q - d + 1$, then the $\epsilon$-error linear complexity is $\LC _\epsilon (\bm S) = \psi (d)$.
			\item If $\epsilon \leq q - \delta_0 + 1=q - d+2 $, then we have
			\[
				\LC _\epsilon(\bm S) \geq \chi_1= -1 + \frac {q\psi(d)-td} {q-d+3} \ge -1 + \frac {1 + (q - d+2) \psi(d)} {q-d+3} = \psi (d ) - 1 - \frac {\psi (d) + 1}{q- d + 3}.
			\]
			\item Let $\varsigma = (d-2)\psi (d) - td$, then $1 \leq \varsigma \leq d-1$.
            If $\epsilon \leq q - \delta_0 + 2=q - d + 3 $, then we have
			\[
				\chi_2 = \psi (d) -\frac {\psi (d) - \varsigma}{q - d+2} \leq \psi(d) - \frac {2 \psi (d) - \varsigma}{q-d+4} = \chi_3.
			\]
			Therefore, the $\epsilon$-error linear complexity of $\bm S$ is
			\[
				\LC_\epsilon(\bm S) \geq -1 + \chi_3 \geq \psi (d) - 1 - \frac {2 \psi (d) - 1}{q-d+4}. \qedhere
			\]
		\end{enumerate}
	\end{proof}

    \begin{rmk}
		By taking $d = 3$, the genus of $M$ is $g (M) = q(q-1)/2$ and there are at least $q^3 + 1$ rational places of $M$ from the proof of \propref{prop:Spl_Cycl_Fun_Fld}. By \cite{RS1994}, $M$ is isomorphic to the Hermitian function field over $\F_{q^2}$. Therefore, the construction above is a generalization of the construction from Hermitian function field \cite[Section III.C]{XD2009}. The bounds of $\epsilon$-error linear complexity in \propref{prop4.6} are slightly better than the results given in \cite[Theorem 3.7]{XD2009}.

	\end{rmk}

	\subsection{Maximal function fields with the second largest genus}

    In this subsection, we provide the explicit construction of  multi-sequences via maximal function fields over $\F_{q^2}$ with the second largest genus. We shall discuss this according to the characteristic of such maximal function fields.

	\subsubsection{The even characteristic case}
		
		\label{subsec:AT_Fun_Fld}

	Let $q =2^s \geq 4$ be a power of $2$. A maximal function field with the second largest genus is defined by
	\[
	Y_2  = \F_{q^2} (x,y), \quad x^{q+1} = \sum_{j = 1}^{s} y^{q/2^j}.
	\]
	To be exact, it is a function field with genus $g(Y_2)  = {q(q-2)} / 4 = \gauss {(q-1)^2 / 4}$, and the set of rational place $\mbb P_{Y_2}^{(1)}$ has cardinality $1+q^3/2$. By \cite[Proposition 2.1]{HM2025}, the common pole $P_\infty$ of functions $x,y \in Y_2$ is a rational place of $Y_2$; besides, for each $\alpha \in \F_{q^2}$, there exists exactly $q/2$ elements $\beta$ satisfying the equation $\alpha^{q+1} = \sum_{j = 1}^{s} \beta^{q/2^j}$, and for such a pair $(\alpha, \beta)$, there exists a unique rational place $P_{\alpha, \beta}$ as the common zero of $x - \alpha, y - \beta $ in $Y_2$. By \cite[Proposition 3.8]{HM2025}, the automorphisms of $Y_2$ over $\F_{q^2}$ are given by
	\[
	\begin{cases}
		\sigma (x) = ax+b, \\ \sigma (y) = y+ (ab^q)^2 x^2 + ab^q x + c,
	\end{cases}
	\]
	where $a,b,c \in \F_{q^2}$ with $a^{q+1} = 1$ and $b^{q+1} = \sum_{j = 1}^{s} c^{q/2^j}$. The construction of a periodic multi-sequence via $Y_2$ is given as follows.

    \begin{thm}
        Let $q \geq 16$ be a power of $2$. Then there exists a multi-sequence over $\F_{q^2} $ of dimension $(q^2- q - 4)  / 4$, period $2(q+1)$ and linear complexity $2(q+1)$.  \label{thm:AT_mult_seq}
    \end{thm}

    \begin{proof}
        Consider the function field $Y_2 / \F_{q^2}$. Let $\zeta$ be a generator of the group $\F_{q^2}^\ast $. Let $0 \neq c \in \F_{q^2}$ satisfy $h(c) = 0$. Let $\sigma\colon Y_2 \to Y_2$ be the $\F_q$-automorphism determined by
        \[
            \sigma (x):= \zeta^{1-q} x, \quad \sigma(y) := y + c.
        \]
        By \cite[Proposition 3.1]{HM2025}, we have that $\sigma (P_\infty) = P_\infty$. It is easy to verify that $\sigma $ is an $\F_{q^2}$-automorphism of order $m = 2 (q+1)$. If $P_{\alpha, \beta} $ is a rational place of $Y_2$ with $\alpha \neq 0$, then the places $P_{\zeta^{j (q-1)}\alpha, \beta + k c} $ for $j=0,1,\dots, q$ and $k =0,1$ are pairwise distinct. Moreover, these places belong to the orbit of $P_{\alpha, \beta} $ under the action by $\langle \sigma \rangle$. Therefore, the set $\{ P_{\alpha, \beta} \in \bbp_ {Y_2} ^ {(1)}\colon \alpha\neq 0 \}$ splits into $((q^3 - q) / 2) / 2(q+1) = (q^2 - q)/4$ pairwise disjoint orbits. Equivalently, there exist $\ell + 1 := (q^2 - q)/4$ rational places of the fixed subfield $Y_2 ^ { \igen {\sigma} } $ split completely in the extension $ Y_2 / Y_ 2 ^ {\igen {\sigma} } $. Let $G = (2g(Y_2 ) - 1) P_\infty$, then $\sigma (G) = G $ and $ \deg (G) \leq (\ell - 1) m - 1$. Therefore, this theorem follows from \thmref{thm:mult_seq_framework}.
    \end{proof}

	\begin{prop}
		Let $q \geq 16$ be a power of $2$. For the multi-sequence $\bm S$ defined in the proof of \thmref{thm:AT_mult_seq}, we have the following results.
		\begin{enumerate}
			\item Its $\epsilon$-error linear complexity is  $2(q+1)$ if $\epsilon \leq (q^2  -2q - 8)/4$;
			\item Its $\epsilon$-error linear complexity is at least $ 2q+1$ if $\epsilon \leq (q^2 - 2q)/4$.
		\end{enumerate}
	\end{prop}
	
	\begin{proof} In the setting of \thmref{thm:general_eps_err}, $\ell = (q^2 - q - 4)/4$, $m = 2q+2$, and $\deg (G) = (q^2 - 2q - 2)/2$. It follows that
		\[
			\delta_0 = 1 + \ceiling{\frac {q(q-2)}{4(q+1)}} = 1 + \ceiling{\frac q4 - \frac {3q}{4q+4}} = 1 + \frac q4.
		\]
        \begin{enumerate}
            \item If $\epsilon \leq \ell -\delta_0 = (q^2 - 2q - 8)/4$, then the $\epsilon$-error linear complexity is $\LC_\epsilon(\bm S) = 2q+2$.
            \item If $\epsilon \leq \ell - \delta_0+1 =  (q^2 - 2q -4)/4$, then \[
			\LC_\epsilon(\bm S) \geq 2q+1 \geq \chi_1 = 2q+1 - \frac {2q+4} {q^2 - 2q}.
		     \]

            If $\epsilon \leq \ell -\delta_0 + 2 = (q^2 - 2q) / 4$, then we have
            \[
                \chi_3 = 2q+2 - \frac {10q + 12} {q^2 - 2q + 4} \leq 2q + 2 - \frac {2q+4} {q^2 - 2q - 4} = \chi_2,
            \]
            Thus, $\LC _\epsilon (\bm S) \geq 2q+1 \geq -1+ \chi_3$. \qedhere
        \end{enumerate}

	\end{proof}
	
	\subsubsection{The odd characteristic case}

    Let $q$ be a power of an odd prime $p$. A maximal function field with the second largest genus can be given by
    \[
        X_2 = \F_{q^2} (x,y), \quad y^q + y = x^{(q+1) / 2}.
    \]
    By \cite[Theorem 3.1]{FGT1996}, it is the unique maximal function field over $\F_ {q^2}$ with the second largest genus, and its genus is $g(X_2) = (q-1)^2 / 4$. By \cite[Proposition 6.4.1 and Example 6.4.2]{GTM254}, $X_2$ has $1 + q(q^2 + 1)/2$ rational places. One of them is the common pole of functions $x,y$ in $X_2$. Other rational places are given as follows: for each $\alpha\in \F_{q^2}$ such that $\alpha^ { (q+1)/2 } \in \F_{q} $, there exist $q$ solutions $\beta$ for the equation $T^q + T = \alpha ^{(q+1)/2} $, and for such a pair $(\alpha, \beta)$, the common zero $P_{\alpha, \beta} $ of functions $x - \alpha, y - \beta$ is a rational place of $X_2$.

    \begin{thm}
        Let $q \geq p^2$ be a power of an odd prime $p$. Then there exists a periodic multi-sequence $\bm S$ over $\F_{q^2}$ of dimension $q(q-1)/p - 1 $, period $p(q+1)/2$ and linear complexity $ p(q+1) /2$. \label{thm:FGT_mult_seq}
    \end{thm}

    \begin{proof}
        Consider the function field $X_2 / \F_{q^2} $. Let $\zeta$ be a generator of the group $\F_{q^2} ^\ast $. Choose $0 \neq c \in \F_{q^2} $ with $c^q + c = 0$. Let $\sigma$ be the variable substitution given by
        \[
            \sigma (x) =  \zeta ^{-2 (q-1) }x, \quad \sigma (y) = y - c.
        \]
        It is easy to verify that $ \sigma(y)^q + \sigma (y) = \sigma (x) ^{(q+1) / 2} $, thus we can extends $\sigma $ to an $\F_{q^2}$-algebra homomorphism $\sigma \colon X_2 \to X_2$. Evidently it is invertible, hence $\sigma \in \Aut (X_2 / \F_{q^2}  )  $. By the direct computation, we can verify that the order of $\sigma $ is $p (q+1) / 2$. For a rational place $P_{\alpha, \beta} $ with $\alpha \neq 0$, it is clear that the places $ P_{\zeta ^{ 2(j - 1) (q-1) } \alpha, \beta + k c} $ for $j = 1,2, \dots, (q+1) / 2  $ and $k = 0,1,\dots, p-1$ are pairwise distinct, and they all belong to the orbit  of $P_{\alpha, \beta} $ under the action by $\langle \sigma \rangle$. Therefore, the set $ \{ P_{\alpha, \beta} \colon \alpha \neq 0 \} $ splits into $( q (q^2 - 1) / 2) / (p(q+1)/2 = q(q-1) / p $ pairwise disjoint orbits. In other words, there exist $ q(q-1) / p$ rational places of the fixed subfield $X_2^{ \igen {\sigma } } $ that split completely in the extension $X_2 / X_2 ^{\igen {\sigma} }  $. Let $G = (2 g(X_2) - 1 ) P_\infty = ((q-1)^2 / 2 - 1) \cdot P_\infty$, then $\sigma (G) = G$. By \thmref{thm:mult_seq_framework}, there exists a multi-sequence $\bm S$ of dimension $\ell = q(q-1)/p - 1 $, period $m = p(q+1)/2$, and linear complexity $\LC(\bm S) = p (q+1)/2$.
    \end{proof}

	\begin{prop}
		Let $\bm S$ be the multi-sequence constructed in the proof \ref{thm:FGT_mult_seq}. The following results hold true.
		\begin{enumerate}
			\item The $\epsilon$-error linear complexity of $\bm S$ is $p(q+1)/2 $ if $\epsilon\leq q(q-2) / p - 2$.
			\item The $\epsilon$-error linear complexity of $\bm S$  is at least $ p(q+1)/2 - 1$ if $\epsilon \leq q (q-2) / p $.
		\end{enumerate}
	\end{prop}
	
	\begin{proof} In the setting of \thmref{thm:general_eps_err}, $\ell = q(q-1)/p - 1$, $m = p(q+1)/2$ and $\deg (G) = (q-1)^2 / 2  - 1$. It follows that \[
		\delta_0 = 1 + \ceiling{\frac {(q-1)^2 / 2} { p (q+1)/2}} = 1 + \ceiling{\frac q p  - \frac {3q-1}{p(q+1)}} =1 + \frac q p.
		\]
        \begin{enumerate}
            \item If $\epsilon \leq \ell -\delta_0 = q (q-2)/p - 2$, then its $\epsilon$-error linear complexity is $\LC_\epsilon (\bm S) = p(q+1)/2$.
		    \item If $\epsilon \leq q (q-2) / p -  1$, then the $\epsilon$-error linear complexity of $\bm S$ is at least
		\[
			\chi_1 = \frac {p (q+1)}2 - 1 - \frac {(p^2 - 3p)q + (p^2 - p)} {2q^2 - 4q}.
		\]
		Thus, $\LC_\epsilon (\bm S) \geq p(q+1) / 2 - 1$.
		
		If $\epsilon \leq q(q-2) / p $, then
        \begin{align*}
            \frac {p(q+1)}2 - \chi_2 &= \frac {(p^2 - 3p)q + (p^2 - p)} {(2q-4)q - 2p} \\
            & \leq \frac {(2p^2 - 3p)q + (2p^2 - p)} {(2q - 4)q + 2p} = \frac {p(q+1)}2 - \chi_3.
        \end{align*}

		It is easy to see that
		\[
			\LC_\epsilon (\bm S) \geq -1 +\ceiling  { \chi_3}=\frac {p (q+1)}2 - 1.
		\]
		Thus, the $\epsilon$-error linear complexity of $\bm S$ satisfies $\LC_{\epsilon}(\bm S) \geq p(q+1)/2 - 1$ if $\epsilon \leq q (q-2) / p $. \qedhere
        \end{enumerate}
	\end{proof}
	
	\label{subsec:FGT_fun_fld}

	\subsection{The Suzuki function field}
	
	\label{sec:Suzuki}
	
	In this subsection, we construct multi-sequences from the Suzuki function field.  Let $q_0$ be a power of $2$ and $q= 2q_0^2$. Let $\F_q$ be the finite field of $q$ elements. The Suzuki function field $ F = \F_q(x,y)$ is the finite algebraic extension of $\F_q(x)$  defined by
	\[
		y^q - y = x^{2q_0} (x ^q  -x).
	\]
	It is a function field of genus $g(F) = q_0 (q-1)$ with a large automorphism group $\Aut (F  /\F_q) = \operatorname{Sz}(q)$ of order $q^2(q^2 + 1)(q-1)$. There are $q^2+1$ rational places of $F$, one of which is the common pole $P_\infty$ of functions $x, y$, and other $q^2$ rational places are exactly common zeros $P_{\mu, \nu} $ of $x- \mu , y- \nu \in F$ where $\mu , \nu  \in \F_q$ are arbitrary \cite{HS1990}. The decomposition group $\Aut (F / \F_q)_\infty$ of $P_\infty$ consists of $\F_{q}$-automorphisms $\sigma$ determined by
	\[
		\begin{cases}
			\sigma (x) = ax + b, \\ \sigma (y) = a^{2q_0+1} y + ab^{2q_0}x + c,
		\end{cases}
	\]
	where $a\in \F_q^*$ and $b,c \in \F_q$ in \cite{BMXY2013}.

	\begin{thm}
		Let $q_0 \geq 8$ be a power of $2$ and $q = 2 q_0^2$. Then there exists a periodic multi-sequence $\bm S$ over $\F_q$ of dimension $q$, period $q-1$ and linear complexity $q-1$.
	\end{thm}

    \begin{proof}
        Consider the Suzuki function field $F= \F_q(x,y)$ defined by $y^q - y = x^{2q_0} (x ^q  -x)$. Let $\zeta$ be a generator of the cyclic group $\F_q^\ast $. Let $\sigma \in \Aut (F / \F_q)_\infty$ be an automorphism determined by
        \[
            \sigma (x)= \zeta x, \quad \sigma (y) = \zeta ^{2q_0 + 1} y.
        \]
        Then the order of $\sigma$ is $q-1$, $\sigma (P_\infty) = P_\infty$ and $\sigma (P_{0,0} ) = P_{0,0} $. By a similar argument as in the proof of \thmref{thm:AT_mult_seq} and \thmref{thm:FGT_mult_seq}, the set $\mathbb{P}_F^{(1)}\setminus \{ P_{0,0}, P_\infty\}$ splits into $q+1$ orbits of length $q-1$ under the action by $\langle \sigma \rangle$. Take $G = (2 g (F) - 1) P_\infty$. By \thmref{thm:mult_seq_framework}, we obtain a multi-sequence $\bm S$ of period $q-1$, dimension $q$ and linear complexity $q-1$.
    \end{proof}

    \begin{prop}
        For the multi-sequence $\bm S$ constructed as above, we have the following results.
        \begin{enumerate}
			\item Its $\epsilon$-error linear complexity is $q - 1$ if $\epsilon \leq q - 2q_0 - 1 $.
			\item Its $\epsilon$-error linear complexity is at least $ q -3 $ if $\epsilon \leq q - 2q_0  $ .
			\item Its $\epsilon$-error linear complexity is at least $ q -4 $ if $\epsilon \leq q - 2q_0+1  $.
		\end{enumerate}
    \end{prop}

	\begin{proof} In the setting of \thmref{thm:general_eps_err}, $\ell = q$, $m = q-1$ and $\deg (G) = 2 q_0 (q-1) - 1$. It follows that
		\[
		\delta_0 = 1 + \ceiling{\frac {2q_0 q - 2q_0}{q-1} } = 2q_0 + 1.
		\]
        \begin{enumerate}
            \item If $\epsilon \leq \ell -\delta_0 = q - 1 -2q_0$, then  $\LC_{\epsilon}(\bm S) = m = q-1$.
            \item If $\epsilon \leq \ell - \delta_0 + 1 = q - 2q_0$, then
		\[
			\LC_\epsilon(\bm S) \geq \ceiling{\chi_1} = \ceiling{ q - 3 - \frac {2q_0 - 3} {q - 2q_0 + 1} } = q-3.
		\]
            \item If $\epsilon \leq \ell - \delta_0 +2 = q - 2q_0 + 1$, then
            \begin{align*}
                q - 3 - \chi_2 \leq 0 \leq \frac {4q_0 - 7} {2 (q_0^2 - q_0 + 1)} = q- 3 - \chi_3.
            \end{align*}
            Therefore, we have
		\[
			\LC_\epsilon (\bm S) \geq \ceiling {-1 + \chi_3} = \ceiling  {q- 4 - \frac {4q_0 - 7} {2 (q_0^2 - q_0 + 1)} } = q-4. \qedhere
		\]
        \end{enumerate}
	\end{proof}

	\subsection{The \GK function field}
	
	\label{sec:GK}
	In this subsection, we investigate the construction of multi-sequences from the maximal function field introduced by Giulietti and Korchm{\'a}ros in \cite{GK2009}.
	
	The \GK  function field is defined by
	\[
		F = \F_{q^6} (x,y,z), \quad \begin{cases}
			z^{q ^2  - q + 1} = y^{q^2} - y, \\
			y^{q+1} = x^q + x.
		\end{cases}
	\]
	It is a maximal function field with genus $g(F) = (q^5 - 2q^3 + q^2)/2$ and $q^8 - q^6 + q^5 + 1$ $\F_{q^6}$-rational places. The $\F_{q^6}$-automorphism group $\Aut (F / \F_{q^6})$ contains a cyclic normal subgroup $C_{q^2-q+1}$ of order $q^2 - q + 1$.
    The cyclic subgroup is known to be
	\[
		C_{q^2 - q+1} := \{(x,y,z) \mapsto (x,y,\mu z) \colon \mu \in \F_{q^6}, \mu^{q^2 - q + 1} = 1  \}.
	\]
	A quick description of rational places of the \GK function field can be found in {\cite[Lemma 4.1]{BMZ2018}}.
	\begin{lem}
		For $c \in \F_{q^6}$, let
		\[
			F_c = \{ (x_0, y_0, z_0) \in \overline {\F}_q \colon x_0^q + x_0 = y_0^{q+1}, y_0^{q^2} - y _0= c^{q^2 - q + 1} \},
		\] where $\overline {\F}_q$ is an algebraic closure of $\F_q$,
		and let
		\[
			\Gamma = \{ c\in \F_{q^6} \colon c^{(q^3 + 1) (q^2 - 1)} + c^{(q^3 + 1) (q^2 - q)} + 1 = 0  \}.
		\]
		Then $F_c$ contains exactly $q^3$ rational places if and only if $c \in \Gamma \cup \{0\}$. The number of distinct sets $F_c$ is $q^5 - q^3 + q^2$, and these distinct sets form a partition of $\mbb P_F^{(1)}  \setminus \{P_\infty\}$, where $P_\infty$ is the common pole of $x,y$ and $z$.
	\end{lem}
	For $\alpha, \beta, \gamma \in \F_{q^6}$ that $\alpha^q + \alpha = \beta^{q+1}$ and $\gamma^{q^2 -  q + 1} = \beta^{q^2} - \beta$, let $P_{\alpha,\beta, \gamma}$ be the common zero of functions $x - \alpha, y - \beta$ and $z - \gamma$.

	\begin{thm}\label{thm4.15}
		Let $q \geq 3$ be a power of prime, There exists a multi-sequence over $\F_{q^6}$ of dimension $q^3 (q+1)^2 (q-1) - 1$, period $q^2 - q + 1$ and linear complexity $q^2 - q + 1$ as well.
	\end{thm}

    \begin{proof}
        Consider the \GK function field $F / \F_{q^6}$, and let $\zeta$ be a primitive $(q^2 - q + 1)$-th root of unity in $\F_{q^6}$. Then the group $C_{q^2 - q + 1}$ is generated by the $\F_{q^6}$-automorphism $\sigma \colon (x,y,z) \mapsto (x,y,\zeta^{-1} z)$. For any rational place $P_{\alpha, \beta, \gamma} $, we have $\sigma (P_{\alpha, \beta, \gamma }) = P_{\alpha, \beta, \zeta \gamma} $. Similarly as in the proof of \thmref{thm:AT_mult_seq}, the set $\{ P_{\alpha, \beta, \gamma} \colon \gamma \neq 0 \} $ splits into $ (q^8 - q^6 + q^5 - q^3)/ (q^2 - q+ 1) = q^3 (q+1)^2 (q-1) $ orbits with length $q^2-q+1$ under the action by $C_{q^2 - q + 1}$. Therefore, there exist $ \ell + 1 = q^3 (q+1)^2 (q-1)$ rational places in the fixed subfield $F^{C_{q^2 - q+ 1} } $, each of which splits completely in the field extension $F / F ^{C_{q^2 - q +1} }$. Put $G = (2 g(F) - 1) P_\infty$. Then $\sigma (G) = G$ and $\deg (G) \leq (\ell -1 )m-1 $. This theorem follows from \thmref{thm:mult_seq_framework}.
    \end{proof}

    \begin{prop}
        The multi-sequence $\bm S$ constructed in the proof of \thmref{thm4.15} satisfies the following statements.
		\begin{enumerate}
			\item Its $\epsilon$-error linear complexity is $q^2 - q + 1$ if $\epsilon \leq q^6+q^5-q^4-2q^3-q^2+2q-1 $.
			\item Its $\epsilon$-error linear complexity is at least $q^2 - q$ if $\epsilon \leq q^6+q^5-q^4-2q^3-q^2+2q+1$.
		\end{enumerate}
    \end{prop}
	
	\begin{proof} In the setting of \thmref{thm:general_eps_err}, $\ell =q^3 (q+1)^2 (q-1) - 1$, $m = q^2 - q+1$ and $\deg (G) = q^5 - 2q^3 + q^2 - 1$. It follows that
	\[
		\delta_0 = 1 + \ceiling{q^3 + q^2 - 2q - 2 + \frac 2 {q^2 - q + 1}} = q^3 + q^2 - 2q.
	\]
    \begin{enumerate}
        \item If $\epsilon \leq \ell -\delta_0 = q^3 (q+1)^2 (q-1) - (q^3 + q^2 - 2q)-1$, then $\LC_\epsilon (\bm S) = q^2- q+1$.
        \item If $\epsilon \leq \ell - \delta_0 + 1=q^6+q^5-q^4-2q^3-q^2+2q$, then
		\[
			\LC_\epsilon(\bm S) \geq  \ceiling {\chi_1}=\ceiling { q^2 - q - \frac 1{\ell - \delta_0 + 2} }=q^2 - q.
		\]

        If $\epsilon\leq \ell -\delta _0 + 2$, then we obtain
        \begin{align*}
            -1 + \chi_2 = q^2 - q - \frac {1} {\ell - \delta_0 + 1} \ \text{ and } \ -1 + \chi_3 = q^2 - q - \frac {q^2 - q + 2}{\ell - \delta_0 + 3}.
        \end{align*}
        Thus, we have $\LC _\epsilon (\bm S) \geq q^2 - q$ by \thmref{thm:general_eps_err}.
    \end{enumerate}
	\end{proof}

	\section{Explicit construction of quasi-cyclic codes via function fields}	
	\label{sec:QC_codes}

	In this section, we exploit the similarity of \QC codes and periodic multi-sequences, and hence provide many explicit \QC codes from various function fields in \secref{sec:instances} as well.  
	

    Let $\F_q$ be the finite field of $q$ elements. For a positive integer $n$, let $\F_q^n$ denote the linear space of dimension $n$ over $\F_q$.  For a vector $\bm a = (a_1, \dots, a_n) \in \F_q^n$ of length $n$, we introduce an operator defined by $$\mscr T \bm a:= (a_2, a_3, \dots, a_{n}, a_1).$$ For a linear code $C \subseteq \F_q^n$, if every codeword $\bm c \in C$ satisfies $\mscr T \bm c \in C$, then we call $C$ a cyclic code. Generally, if there exists a positive integer $\ell \geq 1$ such that
	\[
	\bm c \in C \implies \mscr T^\ell \bm c \in C,
	\]
	then $C$ is called a quasi-cyclic code. Such an $\ell$ is called the index of the \QC code.
     In particular, if $\ell = 1$, then $C$ is cyclic.
    It is easy to see that for \QC codes of length $n$ and index $\ell$, we must have $\ell \mid n$. We write $n = \ell m$. For convenience, we write a codeword $\bm c$ of the \QC code $C$ as
	\[
	\bm c = (c_{1,1}, c_{1,2}, \dots, c_{1,\ell}, c_{2,1}, \dots, c_{2,\ell}, \dots,  c_{m, 1}, c_{m, 2}, \dots, c_{m, \ell}).
	\]
	
	Due to the cyclicity for \QC codes, we can construct \QC algebraic geometric codes under the assumption of \thmref{thm:mult_seq_framework}.

		\begin{thm}\label{thm5.1}
		Let $F /\F_q$ be a function field of genus $g$. Let $m, \ell $ be two positive integers. Assume the following conditions hold true:
		\begin{enumerate}
			\item there exists an $\F_q$-automorphism $\sigma \in \Aut (F/ \F_q)$ of order $m$;
			\item $G$ is a divisor of $F$ with $2g - 1\leq \deg (G) < \ell m$ and $\sigma (G) = G$;
			\item there exist pairwise distinct rational places $P_1, \dots, P_\ell$ of $F$ such that $P_{i,j} := \sigma^{i-1} (P_{j})$ are pairwise distinct for $j = 1, \dots, \ell$ and $i = 1, \dots, m$; moreover, all of them do not belong to $\Supp (G)$.
		\end{enumerate}
		For the list
		$\mscr D = \{P_{1,1}, P_{1,2}, \dots, P_{1,\ell}, P_{2,1}, \dots, P_{2,\ell}, \dots, P_{m,1}, \dots P_{m, \ell}\}$,
		the algebraic geometry code $C = C (\mscr D,G)$ associated to $\mscr D$ and $G$ is a $q$-ary $[\ell m, \deg (G) - g + 1, \geq \ell m - \deg (G)]$ quasi-cyclic code of index $\ell$.  \label{thm:framework}
	\end{thm}
	
	The proof of \thmref{thm5.1} is almost the same with \cite[Theorem 2.4]{BDG2026}, so we omit it.
	
	Now we can apply this theorem to function fields considered in \secref{sec:instances}. By considering the cyclotomic function fields, we obtain the following quasi-cyclic codes.
	
		\begin{cor}
		Let $q \geq 4$ be a power of prime. Let $3 \leq d \leq q-2$ be an odd integer, and let $\psi(d) = (q^d + 1) / (q+1)$. Choose an integer $t$ such that $(d - 2)(\psi (d) - 1) \leq td \leq q^d $, then there exists a $q^2$-ary \QC  code of index $q+1$ with parameters $[q^d + 1, td + 1 - (d - 2)(\psi (d) - 1) / 2,\geq  \! q^d + 1 - td]$.
	\end{cor}

	By considering maximal function fields with the second largest genus, we can construct the following \QC code.

	\begin{cor}
		Let $q \geq p^2$ be a power of a prime integer $p$, and let $m$ be a positive integer that divides $p (q+ 1 ) / 2$, and let $k$ be an integer such that $(q-1)^2 / 2 \leq k < q(q^2- 1)/2$. Then there exists a $q^2$-ary \QC code of index $(q^2 - q)/(2m)$ with parameters $[q (q^2 - 1) / 2, k - \gauss{
		(q-1)^2 / 4} + 1, \geq q(q^2-1)/2 - k]$.
	\end{cor}
	
	By considering the Suzuki function field, we can obtain the following \QC code.
	
	\begin{cor}
		Let $q_0\geq 2$ be a power of $2$, and $q = 2 q_0^2$. Let $ k$ be an integer such that $2q_0 - 2q_0 - 1  \leq k \leq q^2 - 2 $. Then there exists a $q$-ary \QC code of index $q+1$ with parameters $[q^2 - 1, k - q_0 (q-1) + 1, \geq\! (q^2 - 1) - k ]$.
	\end{cor}
	
	Finally, by working with the \GK function field, we can construct the following \QC code.
	
	\begin{cor}
		Let $q \geq 3$ be a power of a prime, and $k$ be an integer satisfying $ q^5 - 2q^3 + q^2 - 1 \leq k \leq q^3(q^3+1)(q^2 - 1) - 1$. Then there exists a $q^6$-ary \QC code of index $q^3 (q+1)^2 (q-1)$ with parameters $[q^3(q^3+1)(q^2 - 1), k + 1 - (q^5 - 2q^3 + q^2)/2, \geq \! q^3(q^3+1)(q^2 - 1) - k ]$.
	\end{cor}

\end{document}